\documentclass[letter, 11pt]{article}
\usepackage{amsmath}
\usepackage{amssymb}
\usepackage{amsthm}
\usepackage{multirow,array}
\usepackage{changepage}
\usepackage{graphicx}
\usepackage[top=1in, bottom=1in, left=1in, right=1in, marginparwidth=1.75cm]{geometry}

\newtheorem{prop}{Proposition}

\begin{document}

\title{Token Curated Registries - A Game Theoretic Approach}
\author{Aditya Asgaonkar\footnote{Visiting student from BITS-Pilani, Goa, India; can also be reached at f20150043@goa.bits-pilani.ac.in }, Bhaskar Krishnamachari\\
Autonomous Networks Research Group\\
 Center for Cyber-Physical Systems and the Internet of Things\\
 Viterbi School of Engineering, University of Southern California\\
 \{aa\_192, bkrishna\}@usc.edu}
 \date{September 5, 2018}
\maketitle
\begin{abstract}
Token curated registries (TCRs) have been proposed recently as an approach to create and maintain high quality lists of resources or recommendations in a decentralized manner. Applications range from maintaining registries of web domains for advertising purposes (e.g., adChain) or restaurants, consumer products, etc. The registry is maintained through a combination of candidate applications requiring a token deposit, challenges based on token staking and token-weighted votes with a redistribution of tokens occurring as a consequence of the vote. We present a simplified mathematical model of a TCR and its challenge and voting process analyze it from a game-theoretic perspective. We derive some insights into conditions with respect to the quality of a candidate under which challenges occur, and under which the outcome is reject or accept. We also show that there are conditions under which the outcome may not be entirely predictable in the sense that everyone voting for accept and everyone voting for reject could both be Nash Equilibria outcomes. For such conditions, we also explore when a particular strategy profile may be payoff dominant. We identify ways in which our modeling can be extended and also some implications of our model with respect to the composition of TCRs.
\end{abstract}

\section{Introduction}

Blockchain technologies promise the development of innovative  decentralized applications that do not require a centralized trusted party. One such class of decentralized applications that are emerging are Token Curated Registries (TCRs). TCRs are decentralized applications that run atop a blockchain framework with the goal of maintaining a dynamic list curated by a set of participants. Each TCR has a token associated with it, and the token holders of this TCR are its participants. The goal of the TCR is to operate so as to ensure that the curated list is of ``high quality". For example, when a TCR intended for curating top 100 cities with highest rainfall is created, the hope is that the curated list will contain the ``true" 100 cities with highest rainfall. The decentralized aspect of TCR aims to ensure that individual organizations cannot distort or censor valuable information in their own self interest.

TCRs originated out of the need to curate and maintain information about top websites for distribution of advertisement content, as this problem is especially susceptible to tampering and falsification of information by large actors to promote their own self-interests. The first concrete use case and implementation of TCRs is the adChain Registry, which is a solution to the aforementioned problem. A specification and minimum viable design of TCRs was proposed by Mike Goldin~\cite{goldinTCR1_0} in 2017, and the adChain application~\cite{adChain} deployed and maintained by MetaX has been live since April 2018.

A TCR is a smart contract that contains the logic for token assignment, voting processes, and list entry maintenance. It is designed with incentive schemes to ensure that participants are rewarded for working towards the goal of high quality list creation, and sufficiently deterred from misbehavior. While such incentive schemes can be produced through various combinations of voting processes and list maintenance, we will discuss a particular instance of an incentive scheme that is based on an updated proposal by Goldin~\cite{goldinTCR1_1}.

To our knowledge, there is no prior theoretical work analyzing TCRs. Given that TCRs have very recently been proposed, and the first production implementation is also in its very early stages, there is also very little empirical evidence of the efficacy of incentive schemes that TCRs employ. In this paper, we attempt to model and analyze TCR through game-theoretic tools.

At its core, the idea of a TCR is to incentivize self-interested voters to do the right thing (create a meaningful, high quality registry). Game theory provides tools to evaluate how a given decentralized mechanism works out in the presence of self-interested parties. In particular, the concept of Nash equilibrium is a powerful one, allowing us to characterize ``stable" and ``unstable" sets of actions.

We first build a mathematical model of the payoffs for various actions of participants in TCRs. In particular, we consider the case when a candidate item has been submitted to the registry and identify the payoffs associated with actions of a single potential challenger (to challenge or not) and each voter (to accept or reject the item). Using these payoffs, we first consider the situation assuming a challenge has occurred and identify what combinations of votes are Nash Equilibria. We then consider what a rational potential challenger would do.

We analyze both a special case of a TCR in which a candidate item is being considered after a challenge, by just two voters (where each voter's vote is decisive), as well as a more general case in which there are any number of voters greater than two. We derive the following theoretical results:
\begin{itemize}
\item For two voters, we show that there are two possible scenarios in general - one in which there is a single threshold condition on the candidate quality such that either the outcome is accept beyond the threshold or reject when the quality is below the threshold; and another in which there are two thresholds on the quality such that the outcome is either always reject below one threshold, always accept beyond the second threshold, and a set of cases in between these thresholds, where both accept and reject are stable outcomes.
\item For more than two voters, we show that there is a threshold beyond which there is no challenge and hence the outcome is always to accept, and below that threshold a challenge happens but there are two possible equilibria - everyone accepts or everyone rejects, so that again there is an equilibrium selection problem.
\item Our present analysis does not rule out the possibility that there may be cases where the threshold for accept satisfies a certain ``strict improvement" property, namely that the value of the registry with the accepted item is strictly better than the value of the registry without it. If this property should hold, there is an interesting, somewhat counter-intuitive, implication for the TCR --- the composition of the registry may depend on the particular order in which the candidate items are submitted to it.
\item For the cases where the game that is defined ends up with an ``equilibrium selection'' problem, we further consider a refinement of the Nash Equilibrium concept called the payoff dominant equilibrium, which is a Pareto-optimal Nash equilibrium. For both the $2$-voter and $n$-voter games, we determine a common additional threshold condition on candidate quality that separates ``all voters accept" and ``all voters reject" as the outcomes corresponding to payoff dominant equilibria.
\end{itemize}

\begin{figure}[h]
    \includegraphics[width=\textwidth]{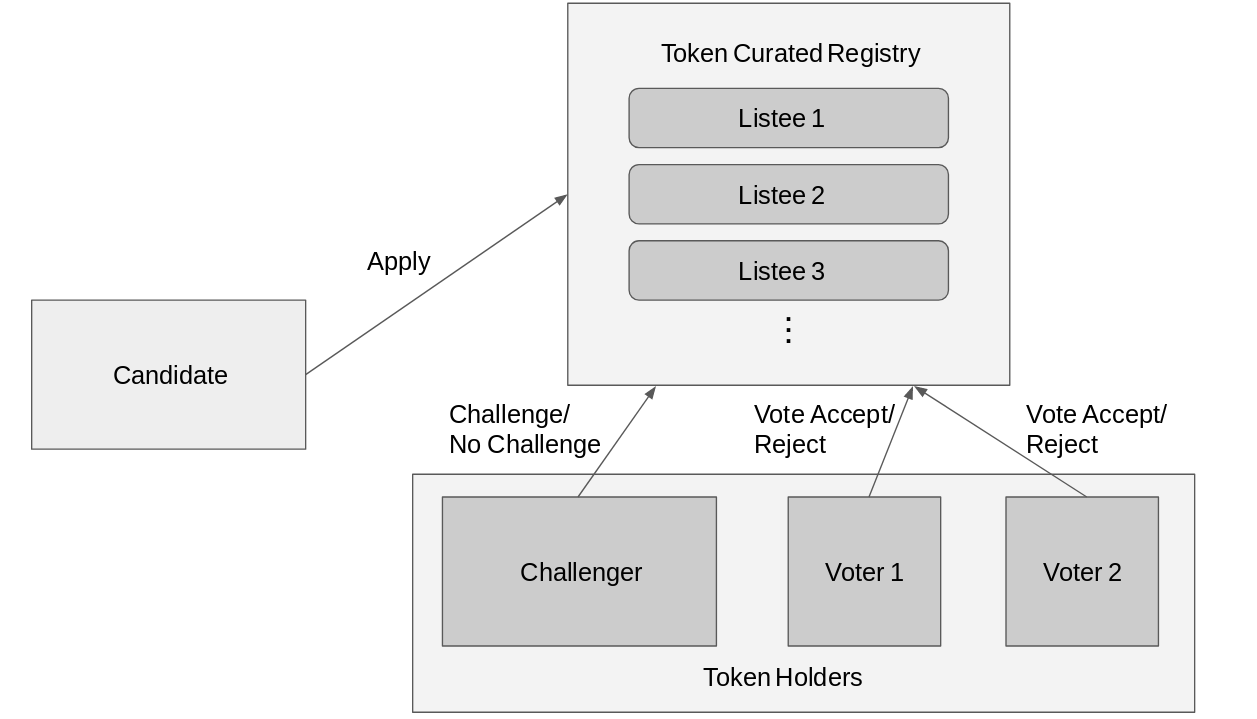}
    \caption{\label{fig:1c-2v-tcr}Illustration of a TCR}
\end{figure}

While it provides some useful insights, due to challenges associated with resolving equilibrium selection problems on a purely theoretical basis, some technical issues with determining a constructive sufficient condition for challenges, and simplifying assumptions made about the quality of a candidate being deterministic and perfectly known to all participants, this study also suggests many open problems. These problems need to be explored further not only from a theoretical perspective, but also through both carefully controlled empirical studies involving human players as well as evaluation of real-world deployments of TCRs.

The rest of the paper is organized as follows. In section~\ref{sec:model}, we formalize the basic model and definitions for the TCR. In section~\ref{sec:games}, we model the payoffs and incentives for various actions and define the voting game. In section~\ref{sec:twovoter}, we analyze the Nash Equilibria and outcomes for a special case of the game involving just two voters, and in section~\ref{sec:nvoter}, we extend our analysis to and analyze outcomes for a more general setting involving any number of $n>2$ voters. Because our analysis identifies settings in which there are multiple equilibria, we further investigate a solution concept for equilibrium selection called Payoff dominant equilibrium in section~\ref{sec:payoffdominance}. We present concluding comments in section~\ref{sec:conclusions}.

\section{Model and Definitions \label{sec:model}}

We follow the TCR 1.1 Model~\cite{goldinTCR1_1} in our formulation. The token curated registry is assumed to be collectively maintained by a set of participants. We focus our attention on the case where there is a single new candidate proposed for the registry who must deposit some tokens, a challenger who is potentially interested in denying entry to the candidate, and a set of voters. If the challenge succeeds, the challenger and those voting to reject the candidate will gain some tokens, while those voting to accept would lose some tokens. If the challenge fails, then the candidate and voters that vote to accept gain some tokens while the challenger and those who voted to reject lose some tokens. If there is no challenge, then the candidate is accepted. It is assumed that the value of the tokens is monotonically increasing in the quality of the list, and hence participants of the TCR (that hold its tokens) have an incentive to increase the quality of the TCR.

Concretely, the participants of interest to us are:
\begin{itemize}
    \item \textbf{Candidate}: This the single actor who is applying to seek entry into the registry.
    \item \textbf{Challenger}: This is the single actor who wants to deny entry to a certain candidate after the application process, possibly because it believes that accepting the candidate will lower the quality of the registry. The act of challenging an application triggers a vote among the token holders to decide whether to permit entry to the candidate.
    \item \textbf{Voters}: These are the token holders of the TCR who vote to decide on permitting new entries in the registry. In our model, the challenger's tokens also get counted towards the reject vote.
\end{itemize}

We consider a Token Curated Registry with the following parameters:
\begin{itemize}
  \item \textbf{Minimum Deposit ($D$)}: This is the minimum deposit that a candidate applying to the TCR must stake. In our model, we assume that the challenger must put down the same amount to challenge the application.
  \item \textbf{Dispensation Percentage ($d$)}: The fraction of the deposit given to the winning challenger/candidate after a challenge.
  \item \textbf{Vote Quorum ($Q$)}: Percentage of total tokens out of total tokens revealed in favor of admitting/keeping a challenged candidate.
  \item \textbf{Minority Bloc Slash ($s$)}: The fraction of the tokens lost by the losing voting bloc, and given to the winning bloc after a challenge.
\end{itemize}

\section{Games, Payoffs and Incentives\label{sec:games}}

Consider the scenario of a candidate applying to the registry by putting down a deposit of $D$ tokens. We are interested in analyzing the game that arises when a token holder (the ``challenger'') challenges this application by putting down a deposit of $D$ tokens. We also subsequently consider whether the potential challenger would even challenge the candidate knowing how that game would play out.

\noindent We use the following notation:
\begin{itemize}
  \item \textbf{Total Voting Tokens ($T$)}
  \item \textbf{Tokens owned by the Challenger ($T_C$)}
  \item \textbf{Total Tokens in favor of Acceptance ($T_A$)}
  \item \textbf{Total Tokens in favor of Rejection ($T_R = T - T_A$)}
  \item \textbf{Ratio of Voting Tokens in favor of Acceptance ($W_A = T_A/T$)}
\end{itemize}

\noindent We also define these additional notations:
\begin{itemize}
  \item \textbf{True Rating of the Candidate ($r$)}
  \item \textbf{Initial Token Valuation ($V(0)$)}: This function describes the valuation of an individual token at the current state of the registry.
  \item \textbf{Token Valuation Function ($V(r)$)}: This function describes the valuation of an individual token if the candidate with given true rating is admitted to the registry. Intuitively, if the candidate was a good choice, then the value of tokens increases ($V(r) > V(0)$). If the candidate was a bad choice, then the value of the tokens decreases ($V(r) < V(0)$).
\end{itemize}

A key assumption of our work is that $V(0)$ and $V(r)$ are deterministic quantities perfectly known to all parties. In future, we plan to relax this assumption.

\noindent At the end of the voting process, the following outcomes are possible:
\begin{itemize}
  \item \textbf{Candidate Accepted}: If $W_A$ (the ratio of weight of the tokens voting in favor of admitting the candidate to the total voting tokens) is greater than (or equals) $Q$, then the candidate is accepted into the list. \\
  In this case, the candidate gets $D \cdot d$ tokens back. The voters who voted in favor of accepting the candidate get $D \cdot (1-d) + (T_R - T_C) \cdot s$ split according to token weight among voters in favor of acceptance. The challenger loses its deposit of $D$ tokens, and all voters in the minority voting bloc lose $s$ fraction of their tokens.
  \item \textbf{Candidate Rejected}: If $W_A$ (the ratio of weight of the tokens voting in favor of admitting the candidate to the total voting tokens) is lesser than $Q$, then the candidate's application is rejected. \\
  In this case, the challenger gets $D \cdot d$ tokens back. The voters who voted in favor of rejecting the candidate get $D \cdot (1-d) + T_A \cdot s$ split according to token weight among voters in favor of rejection. The candidate loses its deposit of $D$ tokens, and all voters in the minority voting bloc lose $s$ fraction of their tokens.
\end{itemize}

\subsection{Perspective of the Candidate}
We assume that the candidate owns exactly $D$ tokens. The candidate has only one option:
\begin{itemize}
  \item \textbf{Apply to Registry}: The candidate chooses to apply to the registry by depositing $D$ tokens.
\end{itemize}

\noindent The two possible outcomes for the candidate are:
\begin{itemize}
  \item \textbf{Candidate Accepted}: If $W_A \geq Q$, then the candidate gets accepted as a listee, and the candidate gains $D \cdot d$ tokens. Total payoff is
  $$
  [D + D \cdot d] \cdot V(r)
  $$
  \item \textbf{Candidate Rejected}: If $W_A < Q$, then the candidate's application is rejected, and the candidate loses $D$ tokens. Total payoff is
  $$
  -D \cdot V(0)
  $$
\end{itemize}

\subsection{Perspective of the Challenger}
We assume that the challenger owns $t ~ (\geq D)$ tokens. The challenger has only one option if should proceed to challenge (if it does not then the candidate is automatically accepted):
\begin{itemize}
  \item \textbf{Challenge}: The challenger chooses to challenge the candidate's application by depositing $D$ tokens.
\end{itemize}

\noindent The two possible outcomes for the challenger are:
\begin{itemize}
  \item \textbf{Candidate Accepted}: If $W_A \geq Q$, then the candidate gets accepted as a listee, and the challenger loses its deposit of $D$. Total payoff is
  $$
  (t - D) \cdot V(r)
  $$
  \item \textbf{Candidate Rejected}: If $W_A < Q$, then the candidate's application is rejected, and the challenger gets $D \cdot d$ tokens. Total payoff is
  $$
  [t + D \cdot d] \cdot V(0)
  $$
\end{itemize}

\subsubsection{Incentive to be a Challenger}
\noindent \textbf{A necessary condition}: a token holder may be incentivized to be a challenger if the maximum payoff from making a challenge is greater than the payoff from not challenging at all, i.e,
$$
max((t - D) \cdot V(r), (t + D \cdot d) \cdot V(0)) > t \cdot V(r)
$$

Two cases arise:
\begin{itemize}
  \item If winning the challenge has better payoff: $(t - D) \cdot V(r) < (t + D \cdot d) \cdot V(0)$. We need:
  \begin{align*}
    (t + D \cdot d) \cdot V(0) &> t \cdot V(r) \\
    \Longleftrightarrow \frac{V(r)}{V(0)} &< \frac{(t + D \cdot d)}{t} \\
    \Longleftrightarrow \frac{V(r)}{V(0)} &< 1 + \frac{D \cdot d}{t}
  \end{align*}

  We also have:
  \begin{align*}
    (t - D) \cdot V(r) &< (t + D \cdot d) \cdot V(0) \\
    \Longleftrightarrow \frac{V(r)}{V(0)} &< \frac{t + D \cdot d}{t - D} \\
    \Longleftrightarrow \frac{V(r)}{V(0)} &< \frac{t - D + D \cdot (1 + d)}{t - D} \\
    \Longleftrightarrow \frac{V(r)}{V(0)} &< 1 + \frac{D \cdot (1 + d)}{t - D} \\
  \end{align*}

  The condition becomes:
  \begin{align*}
  \frac{V(r)}{V(0)} &< min(1 + \frac{D \cdot (1 + d)}{t - D}, 1 + \frac{D \cdot d}{t}) \\
  \Longleftrightarrow \frac{V(r)}{V(0)} &< 1 + min(\frac{D \cdot (1 + d)}{t - D}, \frac{D \cdot d}{t})
  \end{align*}

  Comparing the two values:
  \begin{align*}
  &\frac{D \cdot (1 + d)}{t - D} > \frac{D \cdot d}{t} \\
  \Longleftrightarrow~& \frac{(1 + d)}{t - D} > \frac{d}{t} \\
  \Longleftrightarrow~& t \cdot (1 + d) > (t - D) \cdot d \\
  \Longleftrightarrow~& t + t \cdot d > t \cdot d - D \cdot d \\
  \Longleftrightarrow~& t > - D \cdot d
  \end{align*}

  Let us denote $\frac{D \cdot d}{t}$ by $\delta$.

  Therefore, we need:
  $$
  \frac{V(r)}{V(0)} < 1 + \delta
  $$

  which means that the challenger may be incentivized to challenge, unless $V(r)$ is \emph{sufficiently} better than $V(0)$.

  \item If losing the challenge has better payoff: $(t - D) \cdot V(r) \geq (t + D \cdot d) \cdot V(0)$.
  We have:
  $$
  t \cdot V(r) \geq (t - D) \cdot V(r) \geq (t + D \cdot d) \cdot V(0)
  $$
  which means that the challenger is better off not making a challenge (and candidate is accepted).
\end{itemize}

\noindent \textbf{Attempt to derive a sufficient condition}: a token holder would be incentivized to be a challenger if the minimum payoff from making a challenge is greater than the payoff from not challenging at all, i.e,
$$
min((t - D) \cdot V(r), (t + D \cdot d) \cdot V(0)) > t \cdot V(r)
$$

Two cases arise:
\begin{itemize}
  \item If winning the challenge (candidate rejected) has better payoff: $(t - D) \cdot V(r) < (t + D \cdot d) \cdot V(0)$. We need:
  \begin{align*}
    (t - D) \cdot V(r) &> t \cdot V(r) \\
    \Longleftrightarrow (t - D) &> t
  \end{align*}
  which is trivially false.

  \item If losing the challenge (candidate accepted) has better payoff: $(t - D) \cdot V(r) \geq (t + D \cdot d) \cdot V(0)$.
  We have:
  $$
  t \cdot V(r) \geq (t - D) \cdot V(r) \geq (t + D \cdot d) \cdot V(0)
  $$
  which means that the challenger is better off not making a challenge (and candidate is accepted).
\end{itemize}

This analysis does not lead us to a sufficient condition for the challenger to make a challenge. \\


To work around the above difficulties, we instead adopt the following axiomatic approach for a sufficient condition:  \\


\noindent\textbf{Axiom on Sufficient Condition for Challenge:}
If $\frac{V(r)}{V(0)} = x$ is a sufficient condition for challenger to challenge, then all $\frac{V(r)}{V(0)}$ such that $\frac{V(r)}{V(0)} \leq x$ is a sufficient condition for challenger to challenge. \\

\noindent\textbf{Assumption}: There is a lowest upper bound on $\frac{V(r)}{V(0)}$ such that a challenger will challenge. We denote this lowest upper bound as $1 + \delta'$.  \\

Given the necessary condition derived earlier, we know that $\delta' \leq \delta$.
Given the axiom, the above assumption leads to the following \textbf{inequality for  sufficient condition that the challenger always challenges}:
$$
	\frac{V(r)}{V(0)} < 1 + \delta'
$$

Before proceeding, we note that while we have found it necessary for our modeling effort, the above axiomatic approach is not completely satisfactory, as it doesn't specify the value of $\delta'$ or offer a constructive way to compute it numerically.

\subsection{Perspective of a Non-Challenging, Voting Token Holder}
A non-challenging, voting token holder can choose either of these two options:
\begin{itemize}
  \item \textbf{Accept}: The token holder may choose to vote in favor of accepting the candidate into the list. In this case, the token holder's tokens are counted towards the vote quorum $Q$.
  \item \textbf{Reject}: The token holder may choose to reject the candidate's application to be on the list.
\end{itemize}

Consider a non-challenging, voting token holder who holds $t$ tokens.
There are four possible outcomes for this token holder:
\begin{enumerate}
  \item \textbf{Vote in Favor of Acceptance}:
  \begin{enumerate}
    \item \textbf{Candidate Accepted}: The token holder gets $[D \cdot (1-d) + (T_R - T_C) \cdot s] \cdot \frac{t}{T_A}$ tokens. Total payoff is
    $$
    \{t + [D \cdot (1-d) + (T_R - T_C) \cdot s] \cdot \frac{t}{T_A}\} \cdot V(r)
    $$
    \item \textbf{Candidate Rejected}: The token holder loses $s$ fraction of its tokens. Total payoff is
    $$
    t \cdot (1 - s) \cdot V(0)
    $$
  \end{enumerate}
  \item \textbf{Vote in Favor of Rejection}:
  \begin{enumerate}
    \item \textbf{Candidate Accepted}: The token holder loses $s$ fraction of its tokens. Total payoff is
    $$
    t \cdot (1 - s) \cdot V(r)
    $$
    \item \textbf{Candidate Rejected}: The token holder gets $[D \cdot (1-d) + T_A \cdot s] \cdot \frac{t}{T_R - T_C}$ tokens. Total payoff is
    $$
    \{t + [D \cdot (1-d) + T_A \cdot s] \cdot \frac{t}{T_R - T_C}\} \cdot V(0)
    $$
  \end{enumerate}
\end{enumerate}

\subsubsection{Incentive to Vote}
A sufficient condition: a token holder is incentivized to vote if the minimum payoff from making a vote is greater than the payoff from not voting at all.

Four cases arise:
\begin{enumerate}
  \item \textbf{Vote for Accept}:
  \begin{enumerate}
    \item Losing vote (candidate rejected) has better payoff:
    $$
    t \cdot (1 - s) \cdot V(0) < \{t + [D \cdot (1-d) + (T_R - T_C) \cdot s] \cdot \frac{t}{T_A}\} \cdot V(r)
    $$
    We need:
    \begin{align*}
      t \cdot (1 - s) \cdot V(0) &> t \cdot V(0)\\
      \implies (1 - s) &> 1 \\
      \implies s &< 0
    \end{align*}
    \item Winning vote (candidate accepted) has better payoff:
    $$
    t \cdot (1 - s) \cdot V(0) \geq \{t + [D \cdot (1-d) + (T_R - T_C) \cdot s] \cdot \frac{t}{T_A}\} \cdot V(r)
    $$
  \end{enumerate}
  \item \textbf{Vote for Reject}:
  \begin{enumerate}
    \item Losing vote (candidate accepted) has better payoff:
    $$
    t \cdot (1 - s) \cdot V(r) < \{t + [D \cdot (1-d) + T_A \cdot s] \cdot \frac{t}{T_R - T_C}\} \cdot V(0)
    $$
    We need:
    \begin{align*}
      t \cdot (1 - s) \cdot V(r) &> t \cdot V(0)\\
      \Longleftrightarrow \frac{V(r)}{V(0)} &> \frac{1}{(1 - s)}
    \end{align*}
    \item Winning vote (candidate rejected) has better payoff:
    $$
    t \cdot (1 - s) \cdot V(r) \geq \{t + [D \cdot (1-d) + T_A \cdot s] \cdot \frac{t}{T_R - T_C}\} \cdot V(0)
    $$
  \end{enumerate}
\end{enumerate}

\section{Nash Equilibria and outcomes in a $1$-challenger, $2$-voter TCR system\label{sec:twovoter}}

We now shift our focus to a simple 1-challenger, 2-voter TCR system. We assume that all 3 participants have equal tokens $t$. We also assume that $Q = 0.5$ ,i.e., after the challenge has been made, candidate is rejected iff at least one of the voters votes for reject.

\begin{table}[h]
  \begin{adjustwidth}{-2cm}{}
  \setlength{\extrarowheight}{1.5em}
  \begin{tabular}{cc|c|c|}
    & \multicolumn{1}{c}{} & \multicolumn{2}{c}{Voter $1$}\\
    & \multicolumn{1}{c}{} & \multicolumn{1}{c}{Accept}  & \multicolumn{1}{c}{Reject} \\\cline{3-4}
    \multirow{2}*{Voter $2$}  & Accept &
    $[t + \frac{D \cdot (1 - d)}{2}] \cdot V(r) ~, [t + \frac{D \cdot (1 - d)}{2}] \cdot V(r)$ & 
    $t \cdot (1 - s) \cdot V(0) ~, \{t + [D \cdot (1 - d) + t \cdot s]\} \cdot V(0)$ \\\cline{3-4} 
    & Reject &
    $\{t + [D \cdot (1 - d) + t \cdot s]\} \cdot V(0) ~, t \cdot (1 - s) \cdot V(0)$ & 
    $[t + \frac{D \cdot (1 - d)}{2}] \cdot V(0) ~, [t + \frac{D \cdot (1 - d)}{2}] \cdot V(0)$ \\\cline{3-4} 
  \end{tabular}
  \end{adjustwidth}
  \caption{\label{tab:2-voter-payoffs} Payoff matrix for 2 voters with equal weight}
\end{table}

\subsection{Analysis of Nash Equilibrium for Voters}
Assuming that a challenge has already been proposed, we analyze the voters' strategy profiles to identify the Nash Equilibrium:

\begin{itemize}
\item \textbf{(Accept, Accept)}: Due to the symmetry of the payoffs for each player if defecting, this strategy profile is a Nash Equilibrium iff:
\begin{align*}
	[t + \frac{D \cdot (1 - d)}{2}] \cdot V(r) ~~&\geq~~ \{t + [D \cdot (1 - d) + t \cdot s]\} \cdot V(0) \\
    \Longleftrightarrow \frac{V(r)}{V(0)} ~~&\geq~~ \frac{t + D \cdot (1 - d) + t \cdot s}{t + \frac{D \cdot (1 - d)}{2}} \\
    \Longleftrightarrow \frac{V(r)}{V(0)} ~~&\geq~~ \frac{t + \frac{D \cdot (1 - d)}{2} + \frac{D \cdot (1 - d)}{2} + t \cdot s}{t + \frac{D \cdot (1 - d)}{2}} \\
    \Longleftrightarrow \frac{V(r)}{V(0)} ~~&\geq~~ 1 + \frac{\frac{D \cdot (1 - d)}{2} + t \cdot s}{t + \frac{D \cdot (1 - d)}{2}}
\end{align*}

Let us denote $\frac{\frac{D \cdot (1 - d)}{2} + t \cdot s}{t + \frac{D \cdot (1 - d)}{2}}$ by $\mathcal{E}$.

Therefore, (Accept, Accept) is a Nash Equilibrium iff:
$$
	\frac{V(r)}{V(0)} ~~\geq~~ 1 + \mathcal{E}
$$

\item \textbf{(Accept, Reject) and (Reject, Accept)}: It is clear that for the voter playing Accept, switching to Reject is a better choice because:
$$
	t \cdot (1 - s) \cdot V(0) ~<~ t \cdot V(0) ~<~ [t + \frac{D \cdot (1 - d)}{2}] \cdot V(0)
$$

Therefore, both these strategy profiles are not a Nash Equilibrium.

\item \textbf{(Reject, Reject)}: Due to the symmetry of the payoffs for each player if defecting, this strategy profile is a Nash Equilibrium iff:
\begin{align*}
	[t + \frac{D \cdot (1 - d)}{2}] \cdot V(0) ~~&\geq~~ t \cdot (1 - s) \cdot V(0) \\
	\Longleftrightarrow t + \frac{D \cdot (1 - d)}{2} ~~&\geq~~ t \cdot (1 - s) \\
	\Longleftrightarrow t + \frac{D \cdot (1 - d)}{2} ~~&\geq~~ t ~~\geq~~ t \cdot (1 - s)
\end{align*}
which is trivially true.

Therefore, (Reject, Reject) is always a Nash Equilibrium.

\end{itemize}




\subsection{Outcome Analysis}

For simplicity, let us denote $\frac{V(r)}{V(0)}$ by $\gamma$. Depending on the values of $\delta'$ and $\mathcal{E}$, our analysis above implies that we may have either of these 2 scenarios:
\begin{enumerate}
    \begin{figure}[h]
    \includegraphics[width=\textwidth]{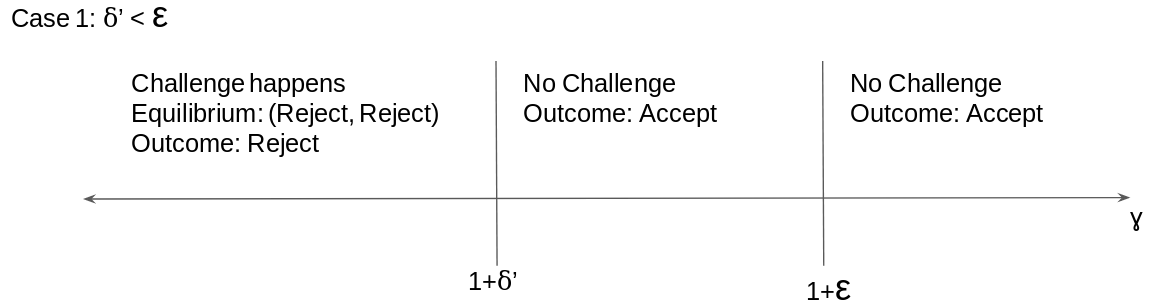}
    \caption{\label{fig:1c-2v-case1}Outcome for $\delta' \leq \mathcal{E}$ in $1$-challenger, $2$-voter system}
    \end{figure}
\item {$\delta' \leq \mathcal{E}$} (Refer Figure \ref{fig:1c-2v-case1}):
	\begin{itemize}
	\item {$\gamma \leq 1 + \delta'$}: Challenge happens and (Reject, Reject) is an equilibrium. Outcome is rejection of candidate.
    \item {$1 + \delta' < \gamma < 1 + \mathcal{E}$}: Challenge does not happen and (Reject, Reject) is an equilibrium. Outcome is acceptance of candidate.
    \item {$1 + \mathcal{E} \leq \gamma$}: Challenge does not happen, and (Reject, Reject) and (Accept, Accept) are both an equilibrium. Outcome is acceptance of candidate.
	\end{itemize}

    \begin{figure}[h]
    \includegraphics[width=\textwidth]{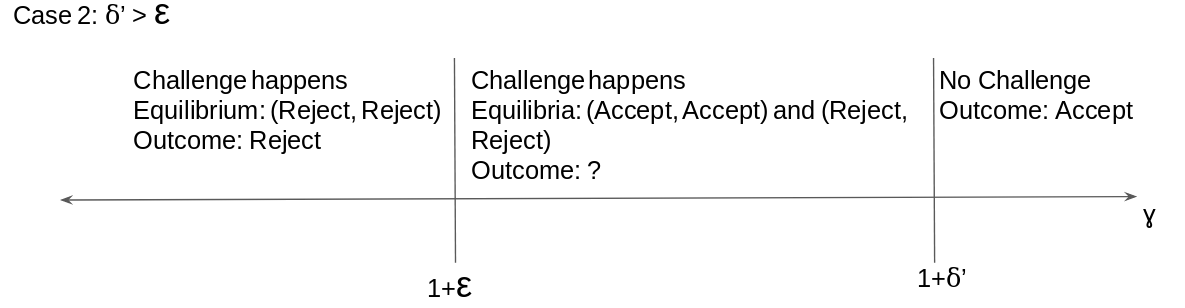}
    \caption{\label{fig:1c-2v-case2}Outcome for $\delta' > \mathcal{E}$ in $1$-challenger, $2$-voter system}
    \end{figure}
\item {$\delta' > \mathcal{E}$} (Refer Figure \ref{fig:1c-2v-case2}):
	\begin{itemize}
	\item {$\gamma \leq 1 + \mathcal{E}$}: Challenge happens and (Reject, Reject) is an equilibrium. Outcome is rejection of candidate.
    \item {$1 + \mathcal{E} < \gamma < 1 + \delta'$}: Challenge happens, and (Reject, Reject) and (Accept, Accept) are both an equilibrium. The outcome is the result of an equilibrium selection problem.
    \item {$1 + \delta' \leq \gamma$}: Challenge does not happen, and (Reject, Reject) and (Accept, Accept) are both an equilibrium. Outcome is acceptance of candidate.
	\end{itemize}

\end{enumerate}

\section{Nash Equilibria and outcomes in a $1$-challenger, $n$-voters TCR system\label{sec:nvoter}}

We now extend our analysis to a more general setting, in which there is still one challenger, and there are $n > 2$ voters.

\begin{prop} So long as a challenge happens, a strategy profile where:
\begin{enumerate}
\item there is at least one vote for accept and one vote for reject, is never a Nash Equilibrium.
\item all voters vote for accept or all voters vote for reject is always a Nash Equilibrium.
\end{enumerate}
\end{prop}
\begin{proof}
    Let us consider the case where $x$ tokens vote for accept, and $T-x$ tokens vote for reject. Consider the case of an individual token holder holding $t$ tokens. The voter can vote in 2 ways:
\begin{enumerate}
\item \textbf{Voting for Accept}:
  \begin{itemize}
  \item \textbf{Candidate Accepted ($x \geq T \cdot Q$)} - Payoff for voting accept is:
  $$
  \{t + [D \cdot (1-d) + (T - x - T_C) \cdot s] \cdot \frac{t}{x}\} \cdot V(r)
  $$
  In case the voter defects to vote reject, then payoff becomes:
  $$
  t \cdot (1 - s) \cdot V(r)
  $$
  \item \textbf{Candidate Rejected ($x < T \cdot Q$)} - Payoff for voting accept is:
  $$
  t \cdot (1 - s) \cdot V(0)
  $$
  In case the voter defects to vote reject, then payoff becomes:
  $$
  \{t + [D \cdot (1-d) + ( x - t ) \cdot s] \cdot \frac{t}{T - ( x - t ) - T_C}\} \cdot V(0)
  $$
  \end{itemize}
\item \textbf{Voting for Reject}:
  \begin{itemize}
  \item \textbf{Candidate Accepted ($x \geq T \cdot Q$)} - Payoff for voting reject is:
  $$
  t \cdot (1 - s) \cdot V(r)
  $$
  In case the voter defects to vote accept, then payoff becomes:
  $$
  \{t + [D \cdot (1-d) + (T - ( x + t ) - T_C) \cdot s] \cdot \frac{t}{x + t}\} \cdot V(r)
  $$
  \item \textbf{Candidate Rejected ($x < T \cdot Q$)} - Payoff for voting reject is:
  $$
  \{t + [D \cdot (1-d) + x \cdot s] \cdot \frac{t}{T - x - T_C}\} \cdot V(0)
  $$
  In case the voter defects to vote accept, then payoff becomes:
  $$
  t \cdot (1 - s) \cdot V(0)
  $$
  \end{itemize}
\end{enumerate}

\noindent To prove part 1 of the proposition, we analyze the following two cases:
\begin{itemize}
\item \textbf{Voter votes Accept, Candidate Rejected} - Voter prefers to defect to voting reject since:
$$
t \cdot (1 - s) \cdot V(0) ~~<~~ t \cdot V(0) ~~<~~ \{t + [D \cdot (1-d) + ( x - t ) \cdot s] \cdot \frac{t}{T - ( x - t ) - T_C}\} \cdot V(0)
$$

\item \textbf{Voter votes Reject, Candidate Accepted} - Voter prefers to defect to voting accept since:
$$
t \cdot (1 - s) \cdot V(r) ~~<~~ t \cdot V(r) ~~<~~ \{t + [D \cdot (1-d) + (T - ( x + t ) - T_C) \cdot s] \cdot \frac{t}{x + t}\} \cdot V(r)
$$
\end{itemize}

One of the above two conditions must happen in part 1 since the vote is not unanimous, there is at least one voter that has an incentive to defect and therefore the given non-unanimous strategy profile is not an equilibrium.
\\
\noindent To prove part 2 of the proposition, we analyze the following two cases:
\begin{itemize}
\item \textbf{Voter votes Accept, Candidate Accepted} - Voter prefers to stick with voting accept since:
$$
\{t + [D \cdot (1-d) + (T - x - T_C) \cdot s] \cdot \frac{t}{x}\} \cdot V(r) ~~>~~ t \cdot V(r) ~~>~~ t \cdot (1 - s) \cdot V(r)
$$

\item \textbf{Voter votes Reject, Candidate Rejected} - Voter prefers to stick with voting reject since:
$$
\{t + [D \cdot (1-d) + x \cdot s] \cdot \frac{t}{T - x - T_C}\} \cdot V(0) ~~>~~ t \cdot V(0) ~~>~~ t \cdot (1 - s) \cdot V(0)
$$
\end{itemize}

In part 2 since the profiles in question consists of all voters voting unanimously one of the above two conditions holds, and in both cases all voters have no incentive to defect unilaterally, thus proving that a unanimous strategy profile is always an equilibrium regardless of whether it consists of all voters voting for accept or reject.   \end{proof}

To summarize, we have that all voters accepting and all voters rejecting are always and the only two equilibria in the $1$-challenger, $n$-voter TCR system so long as a challenge occurs.

\subsection{Outcome Analysis}
Based on the above, we have the following outcomes for various values of $\gamma$ in this case of $n>2$ voters:
\begin{enumerate}
\item \textbf{$\gamma < 1 + \delta'$}: The challenger challenges, and both all voters accepting and all voters rejecting are Nash Equilibria. The outcome is the result of the equilibrium selection problem.
\item \textbf{$\gamma \geq 1 + \delta'$}: The challenger does not challenge (though had there been a challenge, both all voters accepting and all voters rejecting would have been Nash Equilibria). The outcome is that the candidate is accepted.
\end{enumerate}

\section{Ordering and Composition}

As seen in our analyses above, because of the possibility of multiple Nash equilibria, it is not entirely predictable which outcome will prevail in some cases. In particular, if it is possible  that candidates for whom $V(r) > V(0)$ could be rejected  (i.e. candidates that strictly increase the quality of the TCR), then there is an interesting implication for the composition of TCRs: the order in which candidates are proposed to the list can potentially determine the final composition of the list. In particular, it may be possible that the TCR gets more selective over time and lower or intermediate-quality candidates that arrive later might be less likely to be selected than if they arrived when the TCR had fewer (relatively lower quality) elements. As our model doesn't rule out this possibility, it is worth investigating whether this happens in practice.

\section{Payoff Dominant Equilibrium~\label{sec:payoffdominance}}

In both the 2-voter and $n$-voter cases we have found that there can be situations where there are two equilibria: all voters vote to reject or all voters vote to accept the candidate. As one way to deal with the problem of equilibrium selection in games where there can be multiple equilibria, economists have considered a refinement called payoff dominant equilibrium, proposed in the work by Harsanyi and Selten~\cite{harsanyi1988general}. A payoff dominant equilibrium is a Pareto optimal equilibrium, i.e. there is no other equilibrium at which all participants are strictly better off. We now consider and analyze this refinement for the TCR voting games, identifying conditions under which there is a single payoff dominant equilibrium.

\noindent\begin{prop} In the $1$-challenger, $2$-voter TCR system, if there is an equilibrium selection problem and $V(r) > V(0)$, then the payoff dominant equilibrium is (Accept, Accept). Also, if $V(r) < V(0)$, then the payoff dominant equilibrium is (Reject, Reject).\end{prop}

\noindent \begin{proof} A payoff dominant equilibrium is one where there is no other equilibrium where any player stands to gain anything without another player losing something. We analyze by taking an equilibrium in consideration and checking whether any player would lose something if some other equilibrium was played.

When both players play Reject, the payoffs to each player is $\{t + \frac{D \cdot (1 - d)}{2}\} \cdot V(0)$. If both players play Accept, the the payoff to each player is $\{t + \frac{D \cdot (1 - d)}{2}\} \cdot V(r)$. From these payoffs, it is clear that:
$\{t + \frac{D \cdot (1 - d)}{2}\} \cdot V(r) > \{t + \frac{D \cdot (1 - d)}{2}\} \cdot V(0)$ iff $V(r) > V(0)$, i.e., (Accept, Accept) is the payoff dominant equilibrium iff $V(r) > V(0)$.

Similarly, (Reject, Reject) is the payoff dominant equilibrium iff $V(r) < V(0)$.
\end{proof}


We also get a similar result for the $n$-voter case.

\begin{prop} In the $1$-challenger, $n$-voter TCR system, if there is an equilibrium selection problem and $V(r) > V(0)$, then the payoff dominant equilibrium is all voters voting for accept. Also, if $V(r) < V(0)$, then the payoff dominant equilibrium is all voters voting for reject. \end{prop}
\begin{proof} Let us look at the payoffs for a voter with $t$ tokens in both cases:
\begin{itemize}
    \item \textbf{All play accept}: If all voters play accept, then we have number of tokens voting for accept $x = T - T_C$, i.e., all except the challenger voted for Accept. Then the payoff becomes:
    \begin{align*}
    \{t + [D \cdot (1-d) + (T - x - T_C) \cdot s] \cdot \frac{t}{x}\} \cdot V(r) &= \{t + [D \cdot (1-d) + (T - (T - T_C) - T_C) \cdot s] \cdot \frac{t}{(T-T_C)}\} \cdot V(r) \\
    &= \{t + [D \cdot (1-d)] \cdot \frac{t}{(T-T_C)}\} \cdot V(r)
    \end{align*}
    \item \textbf{All play reject}: If all voters play reject, then we have number of tokens voting for accept $x = 0$. Then the payoff becomes:
    \begin{align*}
    \{t + [D \cdot (1-d) + x \cdot s] \cdot \frac{t}{T - x - T_C}\} \cdot V(0) &= \{t + [D \cdot (1-d)] \cdot \frac{t}{T - T_C}\} \cdot V(0) \\
    \end{align*}
\end{itemize}

A payoff dominant equilibrium is one where there is no other equilibrium where any player stands to gain anything without another player losing something. We analyze by taking an equilibrium in consideration and checking whether any player would lose something if some other equilibrium was played.

Therefore, all playing accept is a payoff dominant equilibrium iff:
\begin{align*}
&\{t + [D \cdot (1-d)] \cdot \frac{t}{(T-T_C)}\} \cdot V(r) > \{t + [D \cdot (1-d)] \cdot \frac{t}{T - T_C}\} \cdot V(0) \\
&\Longleftrightarrow V(r) > V(0)
\end{align*}

Similarly, all playing reject is a payoff dominant equilibrium iff $V(r) < V(0)$.  \end{proof}

These findings about the payoff dominant equilibria are of theoretical interest, and  suggest on intuitive grounds that these might be the preferred equilibria. However, we should emphasize that there is no guarantee that rational players will always select this equilibrium, that the other equilibrium cannot arise in a voting game. There are other solution concepts for equilibrium selection that may be worth exploring, such as the concept of a risk dominant equilibrium. Harsanyi, who had first proposed payoff dominant equilibria in a work with Selten in 1988, proposed a theory that is based on risk dominance rather than payoff dominance in a later work published in 1995~\cite{harsanyi1995new}.  And ultimately, the problem of equilibrium selection might be better studied empirically rather than on a purely theoretical basis.

\section{Conclusions\label{sec:conclusions}}

In this paper, we have presented a preliminary mathematical analysis of Token Curated Registry by modeling it game theoretically. In particular, we identified the equilibria associated with different voting patterns as a function of parameters such as the
quality of the candidate being voted upon, and also considered conditions under which a challenge may/may not occur. We have presented the analysis for two sets of cases - one involving just two voters, and one involving any number $n$ of voters.

In general, the following are our findings and conclusions:
\begin{itemize}
\item If the candidate is sufficiently good then in all cases, the outcome is to accept.
\item Likewise if the candidate is sufficiently bad, then we assume (axiomatically) that the outcome is to reject.
\item In between, depending on the TCR parameters, our analysis does not rule out the possibility that the outcome may favor the rejection of a candidate (because of the existence of multiple equilibria). If such a reject outcome were to hold even for candidates with whom the value of the TCR  is strictly better than the quality of the TCR without them, however, it raises a very intriguing possibility -- that the composition of a TCR may depend not only upon the total set of candidates presented to it, but rather the \emph{order} in which the candidates are presented. This surprising possibility is worth further study. 
\end{itemize}

Our work has unearthed some technical challenges with respect to modeling TCR. For one, we have not been able to derive numerically a lowest upper bound on a sufficient condition for a challenger to challenge a candidate, instead assuming it exists on a non-constructive, axiomatic basis. The problem of equilibrium selection is difficult to settle on a purely theoretical basis. While we have investigated the refinement called payoff dominant equilibrium, it's not clear if it will be preferred in practice. In the future, we would like to explore other refinements such as risk dominance, and also conduct experiments with human players or analyze data from real-world deployments of TCR to understand how these games play out in reality.

We have focused on a single-shot game model of the vote for a single candidate.   Another direction for future work is to explicitly explore the dynamics of evolution of a TCR through a suitably-defined repeated game model, although there may be complications induced by the fact that the participants in a TCR themselves might change over time.

In future work, we are also interested in relaxing a key assumption of our model - that the quality of the candidates is deterministic and known to all. In reality, it may be hard for token holders to assess the quality of candidates and that assessment may vary from individual to individual. It may be possible to model this uncertainty using a probabilistic distribution to understand the implications.

\bibliographystyle{plain}
\bibliography{references.bib}

\end{document}